\documentclass{eptcs}

\usepackage{pslatex}
\usepackage{amsfonts}
\usepackage{amsmath}
\usepackage{amsthm}
\usepackage{amssymb}
\usepackage{todonotes}
\usepackage{color}
\usepackage{tikz}

\usetikzlibrary{positioning,calc}

\pgfdeclarelayer{background}
\pgfdeclarelayer{regionbg}
\pgfsetlayers{background,regionbg,main}

\newcommand{\background}[4]
{%
  \begin{pgfonlayer}{background}
    \path (#1.west |- #2.north)+(-0.1,0.1) node [fill=none, minimum height=0pt] (a1) {};
    \path (#3.east |- #4.south)+(+0.1,-0.1) node [fill=none, minimum height=0pt] (a2) {};
    \path[fill=white,rounded corners, draw=black!80, dashed]
    (a1) rectangle (a2);
  \end{pgfonlayer}
}

\newcommand{\titlerunning}{Zenoness for Timed Pushdown Automata}
\newcommand{\authorrunning}{P.A. Abdulla, M.F. Atig, and J. Stenman}

\newtheorem{theorem}{Theorem}
\newtheorem{lemma}[theorem]{Lemma}

\newtheorem{definition}{Definition}

\newcommand{\nats}{\mathbb{N}}

\newcommand{\fun}[3]{#1 \,:\, #2 \to #3}

\newcommand{\nnreals}{\mathbb{R}^{\geq 0}}
\newcommand{\mset}[1]{\lbrace #1 \rbrace}
\newcommand{\msetcond}[2]{\lbrace #1 \, | \, #2 \rbrace}

\newcommand{\clockval}{{\tt X}}

\newcommand{\trans}{\longrightarrow}
\newcommand{\transarg}[1]{\overset{#1}{\longrightarrow}}

\newcommand{\setto}{\leftarrow}

\newcommand{\eps}{\epsilon}
\newcommand{\powersetof}[1]{2^{#1}}

\newcommand{\op}{{\bf op}}
\newcommand{\nop}{{\bf nop}}
\newcommand{\push}{{\bf push}}
\newcommand{\pop}{{\bf pop}}
\newcommand{\test}{{\bf test}}
\newcommand{\reset}{{\bf reset}}

\newcommand{\interval}{\mathcal{I}}
\newcommand{\reference}{\vdash}

\newcommand{\tpda}{\mathcal{T}}
\newcommand{\tpdastates}{Q^{\mathcal{T}}}
\newcommand{\tpdastate}{q}
\newcommand{\tpdainit}{q^{\mathcal{T}}_{init}}
\newcommand{\tpdastack}{\Gamma^{\mathcal{T}}}
\newcommand{\tpdatrans}{\Delta^{\mathcal{T}}}
\newcommand{\tpdaclocks}{X^{\mathcal{T}}}

\newcommand{\pda}{\mathcal{P}}
\newcommand{\pdastates}{Q^{\mathcal{P}}}

\newcommand{\pdainit}{q^{\mathcal{P}}_{init}}
\newcommand{\pdainput}{\Sigma^{\mathcal{P}}}
\newcommand{\pdastack}{\Gamma^{\mathcal{P}}}
\newcommand{\pdatrans}{\Delta^{\mathcal{P}}}

\newcommand{\dur}{\delta}
\newcommand{\confs}{{\it Conf}(\tpda)}
\newcommand{\trule}[1]{\big\langle #1 \big\rangle}
\newcommand{\prefix}{{\mathbf 0}}
\newcommand{\suffix}{{\mathbf 1}}
\newcommand{\cclock}{x_{control}}

\newcommand{\maxconst}{c_{max}}
\newcommand{\transpose}{^\top}

\newcommand{\shadow}[1]{#1^\bullet}
\newcommand{\Max}{Max}
\newcommand{\supports}{\preceq}
\newcommand{\frag}{/}
\newcommand{\ang}[1]{\langle #1 \rangle}

\newcommand{\tmp}{{\tt tmp}}
\newcommand{\tempstates}{{\tt Tmp}}
\newcommand{\floor}[1]{\lfloor #1 \rfloor}

\renewcommand{\emptyset}{\varnothing}

\AtBeginDocument{\RequirePackage{lmodern, times}}

\begin{document}

\pagestyle{headings}
%In order to omit page numbers and running heads
%please change this line to
%\pagestyle{empty}
%and change the first command line too, see above.

%\mainmatter

\title{Zenoness for Timed Pushdown Automata}

\author{Parosh Aziz Abdulla
\and Mohamed Faouzi Atig \and
Jari Stenman}

% \institute{Department of Information Technology \\
% Uppsala University \\
% Sweden \\
% \email{\{parosh, mohamed\_faouzi.atig, jari.stenman\}@it.uu.se}}

\maketitle

%%% Local Variables: 
%%% mode: latex
%%% TeX-master: "main"
%%% End: 

\begin{abstract}
Timed pushdown automata are pushdown automata extended with
a finite set of real-valued clocks. Additionaly,
each symbol in the stack is equipped with a value representing its age.
The enabledness of a transition may depend on the values of the clocks and
the age of the topmost symbol. Therefore, dense-timed pushdown automata subsume
both pushdown automata and timed automata. We have previously shown that the
reachability problem for this model is decidable. In this paper, we study
the zenoness problem and show that it is {\sc ExpTime}-complete.
\end{abstract}

%%% Local Variables: 
%%% mode: latex
%%% TeX-master: "main"
%%% End: 

\section{Introduction}\label{sec:introduction}

%Automata play an important role in the verification of both hardware 
%and software systems.
Pushdown automata \cite{BEM97,Sch02b,EHRS00,ES01} and timed automata \cite{AlurD94,BL-litron08,BCFL04} are two of the most 
widely used models in verification.
Pushdown automata are used as models for (discrete) recursive systems, 
whereas timed automata model timed (nonrecursive) systems. 
Several models have been proposed that extend pushdown automata with timed behaviors \cite{BER94,DIBKS00,Dang03,DBIK04,EmmiM06}.

We consider the model of (Dense-)Timed Pushdown Automata (TPDA),
introduced in \cite{abdulla2012dense}, that subsumes 
both pushdown automata and timed automata. 
As in the case of 
a pushdown automaton, a TPDA has a stack which can be modified 
by pushing and popping.
A TPDA extends pushdown automata with time in
the sense that the automaton (1) has a finite set of real-valued
clocks, and (2) stores with each stack symbol its (real-valued) age.
Pushing a symbol adds it on top of the stack with an
initial age chosen nondeterministically from a given interval.  
A pop transition removes
the topmost symbol from the stack provided that 
it matches the symbol specified by the transition, and  that
its age lies within a given interval.  
A TPDA can also
perform timed transitions, which simulate the passing of time.  
A timed transition synchronously increases the values of all clocks and the
ages of all stack symbols with some non-negative real number.  
The values of the clocks can be tested for inclusion in a given interval or nondeterministically reset
to a value in a given interval.
The model yields a transition system that is infinite in two dimensions;
the stack contains an unbounded number of symbols, and each symbol is
associated with a unique real-valued clock.

In \cite{abdulla2012dense}, we showed that the reachability problem, i.e. the problem of deciding
whether there exists a computation from the initial state to some target state,
is decidable (specifically, {\sc ExpTime}-complete).
In this paper, we address the zenoness problem for TPDA. The zenoness problem is the problem
of deciding whether there is a computation that contains infinitely many discrete transitions
(i.e. transitions that are not timed transitions) in \emph{finite time} \cite{AlurD90,Alur91,Tripakis99}.
Zeno computations may represent specification errors, 
since these kinds of runs are not possible in real-world systems.
We show that the zenoness problem for TPDA can be reduced to the
problem of deciding whether a pushdown automaton has an infinite run
with the labelling $a^{\omega}$.
The latter problem is polynomial in the size of the pushdown automaton,
which is itself exponential in the size of the TPDA.

\subsection*{Related Work}

The works in \cite{BER94,DIBKS00,Dang03,DBIK04,EmmiM06} consider
pushdown automata extended with clocks.
However, these models separate the timed part and the pushdown part of the automaton,
which means that the stack symbols are not
equipped with clocks.
%
%In contrast, we associate with each pushed stack symbol one clock (reflecting its age). 
%%
%In fact, our model can be easily extended such that  each stack symbol
%carries several clocks.

In \cite{BMP10}, the authors define the class of  {\em extended
  pushdown timed automata}. An extended pushdown timed automaton  is a
pushdown automaton enriched with a set of clocks, with an additional
stack used to store/restore clock valuations. 
In our model, clocks  are associated with stack symbols and
store/restore operations are disallowed.
The two models are quite different.
This is illustrated, for instance, by the fact that the reachability problem
is undecidable in their case.

In \cite{ Trivedi:2010}, the authors introduce {\em recursive timed automata},  a model where clocks are considered as variables.  A recursive timed automaton allows passing the values of clocks using either {\em pass-by-value} or {\em pass-by-reference} mechanism. This feature is not supported in our  model since  we do not allow pass-by-value communication between procedures. Moreover, in the recursive timed automaton model, the local clocks of the caller procedure  are stopped until the called procedure returns.
The authors show decidability of the reachability 
problem when either all clocks are passed by reference or none is passed by reference.
This is the model that is most similar to ours, 
since in both cases, the reachability problem reduces to
the same problem for a pushdown automaton that is abstract-time bisimilar
to the timed system.
%

%  %
%  This makes the semantics of the models incomparable with ours, since all the clocks in our model evolve synchronously. 
%  %
% In fact, the authors show  decidability of the reachability 
% problem only in the special cases where
% either all clocks are passed by reference or none is passed by reference.

In a recent work
\cite{PFJ12}
 we have shown decidability
of the reachability problem for
{\it discrete-timed} pushdown automata, where time is interpreted as
being incremented in discrete steps and thus the ages of clocks
and stack symbols are in the natural numbers.
This makes the reachability problem much simpler to solve, and the
method of  \cite{PFJ12} cannot be extended to the dense-time case.

Finally, the zenoness problem for different kinds of timed systems is
well studied in the literature (see, e.g., \cite{Alur91,herbreteau2012efficient} for timed automata
and \cite{abdulla2005decidability} for dense-timed Petri nets).

%\subsubsection{Related Work}
%As mentioned  above several extensions of pushdown automata with  time have been proposed  \cite{BER94,DIBKS00,Dang03,DBIK04,EmmiM06}. In this work we consider the model of timed pushdown automaton proposed in \cite{abdulla2012dense}.

%\subsubsection{Organization}
%Organization of the rest of the paper.
% Timed pushdown automata with only global clocks in [these papers].
% We also have recursive timed automata, where clocks are considered variables.
% They can be passed by reference or by value. 
% The local clocks are stopped until caller procedure returns.

% Another model is extended pushdown automata [ref].
% Blablablablabla.

% We have recently [ref] shown that the discrete case is decidable.

% Zenoness for Timed Automata and Timed Petri Nets is decidable, as shown in [ref].
%%% Local Variables: 
%%% mode: latex
%%% TeX-master: "main"
%%% End: 

%%% Local Variables: 
%%% mode: latex
%%% TeX-master: "main"
%%% End: 

\section{Preliminaries}\label{sec:preliminariesx}

We use $\nats$  and $\nnreals$ to denote  the set of natural numbers  and non-negative reals, respectively.
For values $n, m \in \nats$, we denote by the intervals 
$[n:m]$, $(n:m)$ $[n:m)$, $(n:m]$, $[n: \infty)$ and $(n: \infty)$
the sets of values $r \in \nnreals$ satisfying the constraints
$n \leq r \leq m$, $n < r < m$, $n \leq r < m$, $n < r \leq m$, $n \leq r$,
and $n < r$, respectively.
We let $\interval$ denote the set of all such intervals.

For a non-negative real number $r \in \nnreals$, with $r = n + r'$ 
$n \in \nats$, and  $r' \in [0:1)$, we let $\floor{r} = n$ denote the 
\emph{integral part}, and ${\it frac}(r) = r'$ denote the \emph{fractional part} of $r$.
Given a set $S$, we use $\powersetof{S}$ for the powerset of $S$.
For sets $A$ and $B$, $\fun{f}{A}{B}$ denotes a (possibly partial) function 
from $A$ to $B$. 
We write $f(a) = \bot$ when $f$ is undefined at $a \in A$.
We use $dom(f)$ and $range(f)$ to denote the domain
and range of $f$. We write $f[a \setto b]$ to denote the function
$f'$ such that $f'(a) = b$ and $f'(x) = f(x)$ for $x \not = a$.
The set of partial functions from $A$ to $B$ is written as $[A \to B]$.

Let $A$ be an alphabet. We denote by $A^*$, (resp. $A^+$) the set of all \emph{words}
(resp. non-empty words) over $A$. The empty word is denoted by $\eps$.
For a word $w$, $|w|$ denotes the length of $w$ (we have $|\eps| = 0$).
For words $w_1, w_2$, we use $w_1 \cdot w_2$ for the concatenation of $w_1$ and $w_2$.
We extend the operation $\cdot$ to sets $W_1, W_2$ of words 
by defining $W_1 \cdot W_2 = \msetcond{w_1 \cdot w_2}{w_1\in W_1, w_2 \in W_2}$.
%For a word $w = a_1 \dots a_n$, and $i \in \oneto{n}$, we let $t[i]$ denote $a_i$.
%Given a word $w=(x_1, y_1)\dots (x_n, y_n) \in (X \times Y)^*$,
%we define the \emph{first projection} $\firstproj{t} = x_1 \dots x_n$ 
%and the \emph{second projection} $\secondproj{t} = y_1 \dots y_n$.
We denote by $w[i]$ the $i$th element $a_i$ of $w = a_1 \dots a_n$.

We use $A^{\omega}$ to denote the set of all infinite words over the alphabet $A$.
We let $a^\omega$ denote the infinite word $aaa\dots$ and write $|w| = \infty$ for any infinite word $w$Ê  over $A$.

We define a binary \emph{shuffle operation} $\otimes$ inductively: 
For $w \in (\powersetof{A})^*$, define $ w \otimes \eps = \eps \otimes w = \mset{w}$.
For sets $r_1, r_2 \in \powersetof{A}$ and words $w_1, w_2 \in (\powersetof{A})^*$,
define $(r_1 \cdot w_1) \otimes (r_2 \cdot w_2) = 
(r_1 \cdot (w_1 \otimes (r_2 \cdot w_2))) \cup
(r_2 \cdot ((r_1 \cdot w_1) \otimes w_2))) \cup
((r_1 \cup r_2) \cdot (w_1 \otimes w_2))$. 

Let $w=a_1\dots a_m$ and $w' = b_1 \dots b_n$ be words in $A^*$.
An \emph{injection} from $w$ to $w'$ is a partial function 
$\fun{h}{\mset{1, \dots, m}}{\mset{1, \dots, n}}$
that is \emph{strictly monotonic}, i.e. for all $i,j \in \mset{1, \dots, m}$, if $i < j$ 
and $h(i), h(j) \not = \bot$, then $h(i) < h(j)$.
The \emph{fragmentation} $w/h$ of $w$ w.r.t. $h$ is the sequence
$\ang{w_0}a_{i_1}\ang{w_1}a_{i_2} \dots \ang{w_{k-1}}a_{i_k}\ang{w_k}$,
where $dom(h) = \mset{i_1, \dots, i_k}$ and $w = w_0 \cdot a_{i_1} \cdot w_1 \cdot \dots \cdot a_{i_k} \cdot w_k$.
The fragmentation $w'/h$ is the sequence
$\ang{w_0'}b_{j_1}\ang{w_1'}\dots \ang{w_{l-1}'}b_{j_l}\ang{w_l'}$,
where $range(h) = \mset{j_1, \dots, j_l}$ and
$w' = w_0' \cdot b_{j_1} \cdot \dots \cdot b_{i_l} \cdot w_l'$.

\subsubsection*{Pushdown Automata}

A pushdown automaton is a tuple $(Q, q_{init}, \Sigma, \Gamma, \Delta)$,
where $Q$ is a finite set of states, 
$q_{init}$ is an initial state, 
$\Sigma$ is a finite input alphabet,
$\Gamma$ is a finite stack alphabet and $\Delta$ is a set of transition rules
of the form $\trule{q, \sigma, \nop, q'}$, $\trule{q, \sigma, \pop(a), q'}$ or $\trule{q, \sigma, \push(a), q'}$, 
where $q, q' \in Q$, $a \in \Gamma$ and $\sigma \in \Sigma \cup \mset{\eps}$.

A configuration is a pair $(q, w)$, where $q \in Q$ and $w \in \Gamma^*$.
We define $\gamma_{init} = (q_{init}, \eps)$ to be the \emph{initial configuration},
meaning that the automaton starts in the initial state and with an empty stack.
We define a transition relation $\rightarrow$ on the set of configurations in the following way:
Given two configurations $\gamma_1 = (q_1, w_1)$, $\gamma_2 = (q_2, w_2)$ and a
transition rule $t = \trule{q_1, \sigma, \op, q_2} \in \Delta$, we write $\gamma_1 \transarg{t} \gamma_2$
if one of the following conditions is satisfied:

\begin{itemize}
\item $\op = \nop$ and $w_2 = w_1$,
\item $\op = \push(a)$ and $w_2 = a \cdot w_1$,
\item $\op = \pop(a)$ and $w_1 = a \cdot w_2$.
\end{itemize}

For any transition rule $t = \trule{q_1, \sigma, \op, q_2} \in \Delta$, define $\Sigma(t) = \sigma$.
We define $\trans = \cup_{t \in \Delta} \transarg{t}$
and let $\trans^*$ be the reflexive transitive closure of $\trans$.
We say that an infinite word $\sigma_1\sigma_2\sigma_3 \dots \in \Sigma^\omega$
is a \emph{trace} of $\pda$ if there exists configurations $\gamma_1, \gamma_2, \gamma_3, \dots$
such that $\gamma_1=\gamma_{init}$, Ê$\gamma_1 \transarg{t_1} \gamma_2 \transarg{t_2} \gamma_3 \transarg{t_3} \dots$, and 
$\Sigma(t_i) = \sigma_i$ for all $i \in \nats$.
We denote by ${\it Traces}(\pda)$ the set of all traces of $\pda$.

% We say that the automaton accepts a word $\sigma_1\sigma_2 \dots \sigma_n$
% if there exists a sequence of transitions 
% $\gamma_{init} \transarg{t_1} \gamma_1 \transarg{t_2} \dots \transarg{t_n} \gamma_n$
% such that $\gamma_n \in Q_F$ and $\Sigma(t_i)=\sigma_i$ for all $i \in \mset{1, \dots, n}$.
% If a PDA $P$ accepts exactly all words in a set $S$, 
% we say that $S$ is the \emph{language} $\lang(P)$ of $P$.

%%% Local Variables: 
%%% mode: latex
%%% TeX-master: "main"
%%% End: 

%%% Local Variables: 
%%% mode: latex
%%% TeX-master: "main"
%%% End: 

\section{Timed Pushdown Automata}

\subsection*{Syntax}

A \emph{Timed Pushdown Automaton} (TPDA) is a tuple
$\tpda = (\tpdastates, \tpdainit, \tpdaclocks, \tpdastack, \tpdatrans)$.
Here, $\tpdastates$ is a finite set of \emph{states},
$\tpdainit \in \tpdastates$ is an initial state, 
$\tpdaclocks$ is a finite set of \emph{clocks},
$\tpdastack$ is a finite \emph{stack alphabet}
and $\tpdatrans$ is finite set of \emph{transition rules} of the form
$(\tpdastate, \op, \tpdastate')$, where $\tpdastate, \tpdastate' \in \tpdastates$ and $\op$ is one of the following:

\begin{description}
\item[$\nop$] An ``empty'' operation that does not modify the clocks or the stack,
\item[$\push(a, I)$] Pushes $a \in \tpdastack$ to the stack with a (nondeterministic) 
initial age in $I \in \interval$,
\item[$\pop(a, I)$] Pops the topmost symbol if it is $a$ and its age is in $I \in \interval$,
\item[$\test(x, I)$] Tests if the value of $x \in \tpdaclocks$ is within $I \in \interval$,
\item[$\reset(x, I)$] Sets the value of $x \in \tpdaclocks$ (nondeterministically) to some value
in $I \in \interval$.
\end{description}

Intuitively, a transition rule $\trule{\tpdastate, \op, \tpdastate'}$ means
that the automaton is allowed to move from state $\tpdastate$ to state $\tpdastate'$
while performing the operation $\op$. The $\nop$ operation can be used to switch states
without changing the stack or the values of clocks.

\subsection*{Semantics}

The semantics of TPDA is defined by a transition relation over the set of \emph{configurations}.
A configuration is a tuple $(\tpdastate, \clockval, w)$, where $\tpdastate \in \tpdastates$ is a state,
$\fun{\clockval}{\tpdaclocks}{\nnreals}$ is a \emph{clock valuation} which assigns concrete values to clocks,
and $w = (a_1, y_1) \dots (a_n, y_n) \in (\tpdastack \times \nnreals)^*$ is a \emph{stack content}.
In other words, the stack content is a sequence of pairs, each pair consisting of a symbol and its age.
Here, $(a_1, y_1)$ is on the top and $(a_n, y_n)$ is on the bottom of the stack.
Given a TPDA $\tpda$, we denote by $\confs$ the set of all configurations of $\tpda$.

The transition relation consists of two types of transitions; \emph{discrete} transitions, 
which correspond to applications of the transition rules,
and \emph{timed} transitions, which simulate the passing of time.

\paragraph{Discrete Transitions.}

Let $t = (\tpdastate, \op, \tpdastate') \in \tpdatrans$ be a transition rule and let
$\gamma = (\tpdastate, \clockval, w)$ and $\gamma' = (\tpdastate', \clockval', w')$ be configurations.
We have $\gamma \transarg{t} \gamma'$ if one of the following conditions is satisfied:

\begin{itemize}
\item $\op = \nop$, $w' = w$ and $\clockval' = \clockval$,
\item $\op = \push(a, I)$, $w' = (a, v) w$ for some $v \in I$, and $\clockval' = \clockval$,
\item $\op = \pop(a, I)$, $w = (a, v)w'$ for some $v \in I$, and $\clockval' = \clockval$,
\item $\op = \test(x, I)$, $w' = w$, $\clockval' = \clockval$ and $\clockval(x) \in I$,
\item $\op = \reset(x, I)$, $w' = w$, and $\clockval' = \clockval[x \setto v]$ for some $v \in I$.
\end{itemize}
%
%We have $\gamma \dtrans \gamma'$ if $\gamma \dtransarg{t} \gamma'$ for some $t \in \tpdatrans$.

\paragraph{Timed Transitions.}

Let $r \in \nnreals$ be a real number. Given a clock valuation $\clockval$, let $\clockval^{+r}$ be the function
defined by $\clockval^{+r}(x) = \clockval(x) + r$ for all $x \in X$. 
For any stack content $w = (a_1, y_1) \dots (a_n, y_n)$,
let $w^{+r}$ be the stack content $(a_1, y_1+r) \dots (a_n, y_n+r)$.
Let $\gamma = (\tpdastate, \clockval, w)$ and $\gamma' = (\tpdastate', \clockval', w')$ be configurations.
Then $\gamma \transarg{r} \gamma'$ if and only if  $q' = q$, $\clockval' = \clockval^{+r}$ and $w' = w^{+r}$.
%
%We write $\gamma \ttrans \gamma'$ if $\gamma \ttransarg{r} \gamma'$ for some $r \in \nnreals$.

\paragraph{Computations.}

%We define a transition relation $\transarg{\tau}$ labelled 
%with $\tau \in \tpdatrans \times \nnreals$. 
%For configurations $\gamma$ and $\gamma'$, we have $\gamma \transarg{\tau}$ iff 
%one of the following holds:
%
%\begin{itemize}
%\item $\tau = r$ and $\gamma \ttransarg{r} \gamma'$ for some $r \in \nnreals$.
%\item $\tau = t$ and $\gamma \dtransarg{t} \gamma'$ for some $t \in \tpdatrans$
%\end{itemize}

A \emph{computation} (or \emph{run}) $\pi$  is a (finite or infinite) sequence of  the form 
$(\gamma_1,  {\tau_1},  \gamma_2) (\gamma_2, {\tau_2} ,\gamma_3)  \cdots$ (written as $\gamma_1 \transarg{\tau_1} \gamma_2 \transarg{\tau_2} \gamma_3  \cdots$) such that    $\gamma_i \transarg{\tau_i} \gamma_{i+1}$ for all $1\leq i \leq |\pi|$.
For $\tau \in (\tpdatrans \cup  \nnreals)$, we define $Disc(\tau) = 1$ if $\tau \in \tpdatrans$ and 
$Disc(\tau) = 0$ if $\tau \in \nnreals$. Then, 
the number of discrete transitions in $\pi$ is defined as $|\pi|_{disc}=\sum_{i=1}^{|\pi|} Disc(\tau_i)$.
Note that if $|\pi| = \infty$, then it may be the case that $|\pi|_{disc} = \infty$.

In this paper, we will consider the \emph{duration} of transitions.
Given a $\tau \in (\tpdatrans \cup \nnreals)$, the duration $\dur(\tau)$ is defined in the following way:

\begin{itemize}
\item  $\dur(\tau) = 0$ if $\tau \in \tpdatrans$. Discrete transitions have no duration.
\item $\dur(\tau) = \tau$ if $\tau \in \nnreals$.
\end{itemize}

For a computation $\pi$, we define the duration $\dur(\pi)$ to be $\sum_{i=1}^{|\pi|} \dur(\tau_i)$.
If the automaton can perform infinitely many discrete transitions in finite time, it
exhibits a behavior called \emph{zenoness}.

\begin{definition}[Zenoness]
A computation $\pi$ is \emph{zeno} if it contains infinitely many discrete transitions 
and has a finite duration, i.e. 
if $|\pi|_{disc} = \infty$ and $\dur(\pi) \leq c$ for some $c \in \nats$.
$\pi$ is \emph{non-zeno} if it is not zeno.
\end{definition}

The \emph{zenoness problem} is the question whether a given TPDA contains a zeno run starting from the initial configuration:

\begin{definition}[The Zenoness Problem]
  Given a TPDA $\tpda$, decide if there exists a 
  computation $\pi = \gamma_{init} \trans \gamma_1 \trans \gamma_2 \trans \dots$
  from the initial configuration of $\tpda$ such that $\pi$ is zeno.
\end{definition}

Given two computations $\pi = \gamma_1 \transarg{\tau_1} \gamma_2 \transarg{\tau_2} \gamma_3 \transarg{\tau_3} \cdots$
and $\pi'$, we say that $\pi'$ is a \emph{prefix} of $\pi$ if $\pi = \pi'$ or 
$\pi' = \gamma_1 \transarg{\tau_1} \gamma_2 \transarg{\tau_2} \cdots \transarg{\tau_{n-1}} \gamma_n$ for some $1 \leq n$.
We say that $\pi'$ is a \emph{suffix} of $\pi$ if either
$\pi' = \pi$ or $\pi' = \gamma_n \transarg{\tau_n} \gamma_{n+1} \transarg{\tau_{n+1}} \cdots$ for some 
$n \in \nats$.
We define the \emph{concatenation} of a finite computation 
$\pi = \gamma_1 \transarg{\tau_1} \gamma_2 \transarg{\tau_2} \cdots \transarg{\tau_{n-1}} \gamma_n$ with a (finite or infinite)
computation $\pi' = \gamma_1' \transarg{\tau_1'} \gamma_2' \transarg{\tau_2'} \cdots$,
where $\gamma_n = \gamma_1'$,
as $\pi \cdot \pi' = \gamma_1 \transarg{\tau_1} \cdots \transarg{\tau_{n-1}} \gamma_n 
\transarg{\tau_{1}'} \gamma_2' \transarg{\tau_2'} \cdots$.

Let $\pi = \pi_1 \cdot \pi_2$ be a computation. We call the suffix $\pi_2$ a \emph{unit suffix}
if $\dur(\pi) < 1$. The question whether a TPDA $\tpda$ has a zeno run starting from the initial configuration 
can be reduced to the question whether there exists a run from the initial configuration
which contains a zeno unit suffix:

\begin{lemma}\label{lemma:model1}
A TPDA $\tpda$ contains a zeno run iff $\tpda$ contains a run $\pi = \pi_1 \cdot \pi_2$ such that
$\pi_2$ is zeno and $\dur(\pi_2) < 1$.
\end{lemma}

\begin{proof}
 We prove both directions:
  \begin{description}
  \item[\it If:] By the definition of zenoness.
  \item[\it Only if:]  Assume $\pi$ is a zeno run of $\tpda$. 
    Then there exists a smallest $n \in \nats$ such that
    $\dur(\pi) \leq n$. Call it $c$.
    This means that the longest prefix $\pi'$ of $\pi$ for which $\dur(\pi') \leq c - 1$
    contains finitely many discrete transitions. We have that after
    $\pi'$, the next transition in $\pi$ will be a timed transition $\gamma \transarg{r} \gamma'$ for some $r \in \nnreals$,
    and $\dur(\pi') + r > c - 1$. 
    Now, let $\pi_1 = \pi' \cdot \gamma \transarg{r} \gamma'$, and let $\pi_2$ be the remaining suffix in $\pi$.
    We can conclude that $\dur(\pi_2) = c - \dur(\pi_1) < c - (c - 1) = 1$. 
 %   \todo[fancyline]{Is this proof necessary?}
  \end{description}
\end{proof}

In the rest of the paper, we will show how to decide whether $\tpda$ contains a run that has a zeno unit suffix. Intuitively, 
%For this purpose, we will extend a construction presented in \cite{abdulla2012dense}. 
%\todo[fancyline]{Reformulate?}
given a TPDA $\tpda$, we will  construct a pushdown automaton $\pda$
which simulates the behavior of $\tpda$. The pushdown automaton $\pda$ operates in two modes.

Initially, $\pda$ runs in the first mode, in which it simulates the behavior of $\tpda$ exactly as described in \cite{abdulla2012dense}.
While $\pda$ runs in the first mode, all transitions are labelled with $\eps$.
At any time, $\pda$ may guess that it can simulate a unit suffix.
In this case, $\pda$ switches to the second mode, in which it reads symbols from a unary alphabet (say $\mset{a}$)
while simulating discrete transitions of $\tpda$.
%(a feature of the construction is that simlation of timed
%transitions with duration less than 1 does not change the configuration of $\pda$)
The question whether $\tpda$ contains a unit suffix then reduces to
the question whether ${\it Traces}(\pda)$ includes $a^\omega$.

%%% Local Variables: 
%%% mode: latex
%%% TeX-master: "main"
%%% End: 

%%% Local Variables: 
%%% mode: latex
%%% TeX-master: "main"
%%% End: 

\section{Symbolic Encoding}\label{sec:symbolicx}

In this section, we  show how to construct a symbolic PDA $\pda$ that simulates the
behavior of a TPDA $\tpda$. The PDA uses a symbolic \emph{region} encoding to represent the infinitely many
clock valuations of $\tpda$ in a finite way. The notion of regions was introduced in the classical paper
on timed automata \cite{AlurD94}, in which a timed automaton is simulated by a
region automaton 
(a finite-state automaton that encodes the regions in its states). This abstraction relies on the set of clocks 
being fixed and finite. Since a TPDA may in general operate on unboundedly many clocks (the stack is unbounded, and each symbol has an age),
we cannot rely on this abstraction. 
Instead, we use regions of a special form as stack symbols in $\pda$. For each symbol in the stack of $\tpda$,
the stack of $\pda$ contains, at the same position, a region that relates the stack symbol with all clocks.
A problem with this approach is that we might need to record relations between clocks and stack symbols
that lie arbitrarily far apart in the stack.
However, in \cite{abdulla2012dense}, we show that it is enough to enrich the regions in finite way (by recording the relationship between clocks and adjacent stack symbols), thus keeping the stack alphabet of $\pda$ finite.

\subsection*{Regions}
A region is a word over sets, where each set consists of a number of \emph{items}.
There are \emph{plain items}, which represent the values of clocks and the topmost stack symbols.
In addition, this set includes a reference clock $\reference$, which is always 0 except when simulating a pop transition.
Furthermore, we have \emph{shadow items} which record the values of the corresponding plain items 
in the region below. 
Shadow items are used to remember the time that elapses while the plain symbols they represent are not on the top of the stack.

To illustrate this, assume that the region $R_1$ in Figure \ref{fig:example_region} is the topmost region in the stack.
$R_1$ records the integral values and 
the relationships between the clocks $x_1, x_2$, the topmost stack symbol $a$ and the reference clock $\reference$.
It also relates these symbols to the values of $x_1, x_2$, $b$ and $\reference$ in the previous topmost region.
Now, if we simulate the pushing of $c$ with inital age in $[0:1]$, one of the possible resulting regions is $R_2$.
The region $R_2$ uses $\shadow{x_1}$, $\shadow{x_2}$ and $\shadow{\reference}$ to record the previous values of the clocks
(initially, their values are identical to those of their plain counterparts). The value of the previous topmost symbol $a$ 
is recorded in $\shadow{a}$. Finally, the region relates the new topmost stack symbol $c$ with all the previously mentioned symbols.

We define the set $Y = X \cup \Gamma \cup \mset{\reference}$ of plain items and
a corresponding set $\shadow{Y} = \shadow{X} \cup \shadow{\Gamma} \cup \mset{\shadow{\reference}}$
of shadow items. We then define the set of \emph{items} $Z = Y \cup \shadow{Y}$.

Let $\maxconst$ be the largest constant in the definition of $\tpda$. We denote by $\Max$ the set 
$\mset{0, 1, \dots, \maxconst, \infty}$. A \emph{region} $R$ is a word 
$r_1 \dots r_n \in (\powersetof{Z \times \Max})^+$ such that the following holds:

\begin{itemize}
\item $\sum_{i=1}^{n} |(\Gamma \times \Max) \cap r_i| = 1$ 
and $\sum_{i=1}^{n} |(\shadow{\Gamma} \times \Max) \cap r_i| = 1$.
There is exactly one occurrence of a stack symbol and one occurrence of a shadow stack symbol.
\item $\sum_{i=1}^{n} |(\mset{\vdash} \times \Max) \cap r_i| = 1$ 
and $\sum_{i=1}^{n} |(\mset{\shadow{\vdash}} \times \Max) \cap r_i| = 1$.
There is exactly one occurrence of $\vdash$ and one occurrence of $\shadow{\vdash}$.
\item For all clocks $x \in X$, $\sum_{i=1}^{n} |(\mset{x} \times \Max) \cap r_i| = 1$ 
and $\sum_{i=1}^{n} |(\mset{\shadow{x}} \times \Max) \cap r_i| = 1$. 
Each plain clock symbol and shadow clock symbol occurs exactly once.
\item $r_i \not = \emptyset$ for all $2 \leq i \leq n$. Only the first set may be empty.
\end{itemize}

%%% Local Variables: 
%%% mode: latex
%%% TeX-master: "main"
%%% End: 

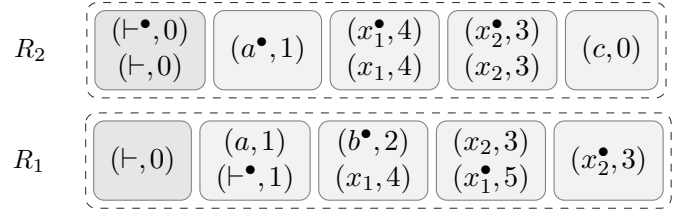
\begin{figure}
\centering

\begin{tikzpicture}
\tikzstyle{letter}=[fill=black!5, minimum width=30pt, minimum height=30pt, rectangle, rounded corners, thin, draw=black!50, inner sep=0pt]
\tikzstyle{separator}=[fill=none, draw=none, minimum height=40pt]
\tikzstyle{zero}=[fill=black!10]
\tikzstyle{reglabel}=[fill=none, draw=none, font=\small]

\node [letter,zero] (l0) 
{
\begin{tabular}{c}
$(\vdash, 0)$ \\
\end{tabular}
};

\node [letter, right=0.1cm of l0] (l1) 
{
\begin{tabular}{c}
$(a, 1)$ \\
$(\shadow{\vdash}, 1)$
\end{tabular}
};

\node [letter, right=0.1cm of l1] (l2) 
{
\begin{tabular}{c}
$(\shadow{b}, 2)$ \\
$(x_1, 4)$
\end{tabular}
};

\node [letter, right=0.1cm of l2] (l3) 
{
\begin{tabular}{c}
$(x_2, 3)$ \\
$(\shadow{x_1}, 5)$
\end{tabular}
};

\node [letter, right=0.1cm of l3] (l4) 
{
\begin{tabular}{c}
$(\shadow{x_2}, 3)$ \\
\end{tabular}
};

\node [letter,zero, above=0.4cm of l0, xshift=1mm] (u0) 
{
\begin{tabular}{c}
$(\shadow{\vdash}, 0)$ \\
$(\vdash, 0)$
\end{tabular}
};

\node [letter, right=0.1cm of u0] (u1) 
{
\begin{tabular}{c}
$(\shadow{a}, 1)$ \\
\end{tabular}
};

\node [letter, right=0.1cm of u1] (u2) 
{
\begin{tabular}{c}
$(\shadow{x_1}, 4)$ \\
$(x_1, 4)$
\end{tabular}
};

\node [letter, right=0.1cm of u2] (u3) 
{
\begin{tabular}{c}
$(\shadow{x_2}, 3)$ \\
$(x_2, 3)$
\end{tabular}
};

\node [letter, right=0.1cm of u3] (u4) 
{
\begin{tabular}{c}
$(c, 0)$
\end{tabular}
};

\background{u0}{u0}{u4}{u0};
\background{l0}{l0}{l4}{l0};

\node [reglabel, left=0.5cm of u0] {$R_2$};
\node [reglabel, left=0.5cm of l0] {$R_1$};

\end{tikzpicture}

\caption{Two examples of regions}
\label{fig:example_region}
\end{figure}

For items $z \in Z$, if we have $(z, k) \in r_i$ for some $i \in \mset{1, \dots, n}$ and some (unique) $k \in \Max$,
then define ${\it Val}(R, z) = k$ and ${\it Index}(R, z) = i$. Otherwise, define ${\it Val}(R, z) = \bot$ and ${\it Index}(R, z) = \bot$ (this may be the case for stack symbols). We define $R\transpose = \msetcond{z \in Z}{{\it Index}(R, z) \not = \bot}$.

%Figure \ref{fig:example_region} shows an example of a region.

\subsection*{Operations on Regions}

In order to define the transition rules of the symbolic PDA, we need a number of operations on regions:

\subsubsection*{Testing Satisfiability}

When we construct new regions, we need to limit the values of the items to certain intervals.
To do this, we define what it means for a region to \emph{satisfy} a membership predicate.
Given an item $z \in Z$, an interval $I \in \interval$ and a region $R$ such that $z \in R\transpose$,
we write $R \vDash (z \in I)$ if and only if one of the following conditions is satisfied:

\begin{itemize}
\item $Index(R, z) = 1$, $Val(R, z) \not= \infty$ and $Val(R, z) \in I$,
\item $Index(R, z) > 1$, $Val(R, z) \not = \infty$ and $Val(R, z) + v \in I$ for all $v \in \nnreals$ such that
$0 < v < 1$,
\item $Val(R, z) = \infty$ and $I$ is of the form $(m:\infty)$ or the form $[m:\infty)$ for some $m \in \nats$.
\end{itemize}

\subsubsection*{Adding and Removing Items}

In the following, we define operations that describe how items are added and deleted from regions.
We also define, in terms of these operations, an operation that assigns a new value to an item.

For a region $R = r_1 \dots r_n$, an item $z \in Z$ and an $k \in \Max$, we define $R \oplus (z, k)$
to be the set of regions $R'$ satisfying the following conditions:

\begin{itemize}
\item $R = r_1 \dots r_{i-1} (r_i \cup \mset{(z, k)}) r_{i+1} \dots r_n$, where $1 \leq i \leq n$
\item $R = r_1 \dots r_i \mset{(z, k)} r_{i+1} \dots r_n$, where $1 \leq i \leq n$
\end{itemize}
We extend the definition of $\oplus$ by letting $R \oplus a$ denote the set $\bigcup_{m \in \Max} R \oplus (a,m)$,
i.e. the set of regions where we have added all possible values of $a$.

We define $R \ominus z$ to be the region $R' = r_1' \dots r_n'$, where, for $1 \leq i \leq n$, 
we have $r_i' = r_i \setminus \mset{\mset{z} \times \Max}$ 
if $r_i \setminus \mset{\mset{z} \times \Max} \not = \emptyset$,
and $r_i' = \eps$ otherwise. 
%We extend the definition of $\ominus$ to sets of items in the following way:
%$R \ominus \mset{z_1, \dots, z_n} = (\dots ((R \ominus z_1) \ominus z_2) \dots) \ominus z_n$
We extend the definition of $\ominus$ to sets of items in the following way:
$R \ominus \emptyset = R$ and $R \ominus \mset{z_1, \dots z_n} = (R \ominus z_1) \ominus \mset{z_2, \dots z_n}$.

Given a region $R$, an item $z \in Z$ and an interval $I \in \interval$, 
we define an \emph{assignment} operation. We write $R[z \setto I]$ 
to mean the set of regions $R'$ such that $R' \in (R \ominus z) \oplus z$ and $R' \vDash (z \in I)$.
For any number $n \in \nats$, we write $R[z \setto n]$ to mean $R[z \setto [n:n]]$.

\subsubsection*{Creating New Regions}

When we push a new stack symbol, 
we need to record the values of clocks and the value of the current top-most stack symbol.
The operation ${\it Make}$ takes as arguments a region, a stack symbol, and an interval, 
It constructs the set of regions in which the shadow items record the values of the plain items in the 
old topmost region, and the value of the stack symbol is in the given interval.

Given a region $R$, a stack symbol $a \in \Gamma$ and an interval $I \in \interval$, we define ${\it Make}(R, a \in I)$
to be the set of regions $R'$ such that there are $R_1, R_2, R_3$ satisfying the following:

\begin{itemize}
\item $R_1 = R \ominus (R\transpose \cap \shadow{Y})$,
\item 
If $R_1 = r_1 \dots r_n$, then $R_2 = r_1' \dots r_n'$, 
where $r_i' = r_i \cup \msetcond{(\shadow{y}, k)}{(y, k) \in r_i}$
for $i \in \mset{1, \dots, n}$,
\item $R_3 = R_2 \ominus (R\transpose \cap \Gamma)$,
\item $R' \in R_3 \oplus a$ and $R' \vDash (a \in I)$.
\end{itemize}

\subsubsection*{Passage of Time}

We implement the passage of time by \emph{rotating} the region. 
A rotation describes the effect of the smallest timed transition
that changes the region.
If the leftmost set (i.e. the set which represents items with fractional part 0) is nonempty,
a timed transition, no matter how small, will ``push'' those items out.
If the leftmost set is empty, the smallest timed transition that changes the regions is one
that makes the fractional parts of those items 0.

Given a pair $(z, k) \in Z \times \Max$, define $(z,k)^+ = (z, k')$, 
where $k' = k+1$ if $k < \maxconst$ and $k' = \infty$ otherwise.
For a set $r \in \powersetof{Z \times \Max}$, define $r^+ = \msetcond{(z, k)^+}{(z, k) \in r}$.
For a region $R = r_1 \dots r_n$, we define $R^+ = R'$ such that one of the following conditions is satisfied:

\begin{itemize}
\item $r_1 \not = \emptyset$ and $R' = \emptyset r_1 \dots r_n$,
\item $r_1 = \emptyset$ and $R' = r_n^+ r_1 \dots r_{n-1}$.
\end{itemize}
We denote by $R^{++}$ the set $\mset{R, R^+, (R^+)^+, ((R^+)^+)^+, \dots}$. Note that this set is finite.

\subsubsection*{Product}

When we simulate a pop transition, the region that we pop contains the most recent values of all clocks.
On the other hand, the region below it contains shadow items that record relationships between items 
further down the stack. We need to keep all of this information. To do this, we 
define a product operation $\odot$ that merges the information contained in two regions.
For regions $P=p_1\dots p_{|P|}$ and $Q = q_1 \dots q_{|Q|}$, and an injection $h$ 
from $\mset{1, \dots, |P|}$ to $\mset{1, \dots, |Q|}$, we write
$P \supports_h Q$ iff the following conditions are satisfied:

\begin{itemize}
\item $Val(P, \shadow{y}) = Val(Q, y)$ for all $y \in P\transpose \cap Y$,
\item For every $i > 1$, $h(i) \not = \bot$ iff there exists a $y \in Y$
such that $Index(P, y) = i$,
\item $h(1) = 1$,
\item For all $y \in Y$, $i \in \mset{1, \dots, |P|}$ and $j \in \mset{1, \dots, |Q|}$,
if $Index(P, y) = i$ and $Index(Q, \shadow{y}) = j$, then $h(i) = j$.
\end{itemize}
We say that $P$ \emph{supports} $Q$, written $P \supports Q$, if $P \supports_h Q$ for some $h$.
Let $P \frag h = p_{i_1}\ang{P_1}p_{i_2} \dots p_{i_m}\ang{P_m}$
and $Q \frag h = q_{j_1}\ang{Q_1}q_{j_2} \dots q_{j_m}\ang{Q_m}$.
We define $p_k' = p_{i_k} \cap (\shadow{Y} \cup \Gamma)$ and
$q_k' = q_{j_k} \cap (X \cup \mset{\vdash})$.
Finally, define $r_1 = p_1' \cup q_1'$ and, for $k \in \mset{2, \dots, m}$, define $r_k = p_k' \cup q_k'$
if $p_k \cup q_k' \not = \emptyset$ and $r_k = \eps$ if $p_k \cup q_k'  = \emptyset$.
Then, $R \in P \odot Q$ if $R = r_1 \cdot R_1 \cdot r_2 \dots r_m \cdot R_m$
and $R_k \in P_k \otimes Q_k$ for $k \in \mset{1, \dots, m}$.

%%% Local Variables: 
%%% mode: latex
%%% TeX-master: "main"
%%% End: 

\section{An {\sc Exptime} Upper Bound for the Zenoness Problem}\label{sec:simulation}

In this section, we prove our main result:

\begin{theorem}\label{thm:exptime}
The Zenoness problem for TPDA is in {\sc ExpTime}.
\end{theorem}

The rest of this section will be devoted to the proof of Theorem \ref{thm:exptime}.
%First, we prove membership in {\sc ExpTime} by reducing the zenoness problem for TPDA 
%to the problem of checking whether $a^\omega$ is contained in the traces of a PDA.
Given a TPDA $\tpda = (\tpdastates, \tpdainit, \tpdastack, \tpdaclocks, \tpdatrans)$,
we construct an (untimed) PDA $\pda = (\pdastates, \pdainit, \pdainput, \pdastack, \pdatrans)$
such that $\pda$ simulates zeno runs of $\tpda$. 
More specifically, $\pda$ simulates a zeno run of $\tpda$
by first simulating the prefix, and then simulating the unit suffix.
In order to do this, $\pda$ runs in two modes. In the first mode, it simulates the prefix.
In the second mode, it simulates the suffix while keeping track of the fact  that the value of a special \emph{control clock}
$\cclock$ is smaller than 1.
We now describe the components of $\pda$.
 
The states of $\pda$ are composed of two disjoint sets; 
the \emph{genuine} states $\mset{\prefix, \suffix} \times \tpdastates$
and some \emph{temporary} states $\tempstates$. 
Each genuine state $(m, q)$ contains a state $q$ from $\tpdastates$ and a 
symbol $m$ indicating the current simulation mode. 
If $m = \prefix$, $\pda$ is currently simulating the prefix of
a run. Conversely, if $m = \suffix$, $\pda$ is simulating the suffix.
The temporary states are  used for intermediate transitions between configurations
containing genuine states.
We assume that we have functions $\tmp$, $\tmp_1$ and $\tmp_2$
that input arguments and map them to  a unique element in $\tempstates$.
The initial state $\pdainit$ of $\pda$ is the state $(\prefix, \tpdainit)$.
The input alphabet $\pdainput$ is the unary alphabet $\mset{a}$.
The automaton reads an $a$ when (and only when) it simulates a discrete transition 
in the suffix. When it simulates any other transition, it reads $\eps$.
Let $\cclock \not \in \tpdaclocks$ be a special control clock. The stack alphabet $\pdastack$ contains all possible 
regions over the items $Z \cup \mset{\cclock, \shadow{\cclock}}$.
The purpose of the control clock is to limit the duration of the suffix.
We will now describe the set $\tpdatrans$ of transition rules:

\paragraph{$\nop$}

For each transition rule $\trule{\tpdastate_1, \nop, \tpdastate_2} \in \tpdatrans$,
the set $\pdatrans$ contains the transition rules 
$\trule{(\prefix, \tpdastate_1), \eps, \nop, (\prefix, \tpdastate_2)}$
and $\trule{(\suffix, \tpdastate_1), a, \nop, (\suffix, \tpdastate_2)}$.
Nop transitions are used for switching states without modifying the clocks or the stack.

\paragraph{$\test(x \in I)$}

We simulate a test transition in $\tpda$ with two transition in $\pda$.
If the topmost region satisfies the constraint, we pop it 
and move to a temporary state. Since a test transition is not supposed to modify the stack,
we push back the same region we popped, while moving to the second genuine state.
Formally, for each transition rule $\tau = \trule{\tpdastate_1, \test(x \in I), \tpdastate_2} \in \tpdatrans$,
and region $R$ such that $R \vDash (x \in I)$, the set $\tpdatrans$ contains the transition rules:

\begin{itemize}
\item $\trule{(\prefix, \tpdastate_1), \eps, \pop(R), \tmp(\tau, R, \prefix)}$,
\item $\trule{\tmp(\tau, R, \prefix), \eps, \push(R), (\prefix, \tpdastate_2)}$ (for simulating the prefix),
\item $\trule{(\suffix, \tpdastate_1), a, \pop(R), \tmp(\tau, R, \suffix)}$,
\item $\trule{\tmp(\tau, R, \suffix), \eps, \push(R), (\suffix, \tpdastate_2)}$ (for simulating the suffix).
\end{itemize}

\paragraph{$\reset(x \setto I)$}

We simulate reset transitions by popping the topmost region and pushing it back, in a similar way to 
test transitions, except that the given clock is nondeterministically set to some value in the given interval.
Formally, for each transition rule $\tau = \trule{\tpdastate_1, \reset(x \setto I), \tpdastate_2} \in \tpdatrans$, and each
pair of regions $R, R'$ such that $R' \in R[x \setto I]$, the set $\tpdatrans$ contains the transition rules:

\begin{itemize}
\item $\trule{(\prefix, \tpdastate_1), \eps, \pop(R), \tmp(\tau, R, \prefix)}$,
\item $\trule{\tmp(\tau, R, \prefix), \eps, \push(R'), (\prefix, \tpdastate_2)}$ (for simulating the prefix),
\item $\trule{(\suffix, \tpdastate_1), a, \pop(R), \tmp(\tau, R, \suffix)}$,
\item $\trule{\tmp(\tau, R, \suffix), \eps, \push(R'), (\suffix, \tpdastate_2)}$ (for simulating the suffix).
\end{itemize}

\paragraph{$\push(a, I)$} 

We will need two temporary states to simulate a push.
First, we move to a temporary state while popping the topmost region.
This is done in order to remember its content. Then, we push back that region
unmodified. Finally, we push a region containing the given symbol, 
constructed from the previous topmost region such that the initial age of the symbol is in the given
interval. 
Formally, for each transition rule $\tau = \trule{\tpdastate_1, \push(a, I), \tpdastate_2} \in \tpdatrans$, and each
pair of regions $R, R'$ such that $R' \in Make(R, a \in I)$, the set $\tpdatrans$ contains the transition rules:

\begin{itemize}
\item $\trule{(\prefix, \tpdastate_1), \eps, \pop(R), \tmp_1(\tau, R, \prefix)}$,
\item $\trule{\tmp_1(\tau, R, \prefix), \eps, \push(R), \tmp_2(\tau, R, \prefix)}$,
\item $\trule{\tmp_2(\tau, R, \prefix), \eps, \push(R'), (\prefix, \tpdastate_2)}$ (for simulating the prefix),
\item $\trule{(\suffix, \tpdastate_1), a, \pop(R), \tmp(\tau, R, \suffix)}$,
\item $\trule{\tmp_1(\tau, R, \suffix), \eps, \push(R), \tmp_2(\tau, R, \suffix)}$,
\item $\trule{\tmp(\tau, R, \suffix), \eps, \push(R'), (\suffix, \tpdastate_2)}$ (for simulating the suffix).
\end{itemize}

\paragraph{$\pop(a, I)$} 

The simulation of pop transitions also requires two temporary states.
First, we pop the topmost region and move to a temporary state.
Then, in order to update the new topmost region, we need to first pop it, then rotate and merge it with the
first region we popped, and finally push back the result.
Formally, for each transition rule $\tau = \trule{\tpdastate_1, \pop(a, I), \tpdastate_2} \in \tpdatrans$, and all 
regions $R_1, R_1', R_2,$ such that $R_2 \vDash (a \in I)$ 
and $R_1' \in  \bigcup \msetcond{R_2 \odot R'}{R' \in R_1^{++} \text{ and } R' \supports R_2}$,
the set $\tpdatrans$ contains the transition rules:

\begin{itemize}
\item $\trule{(\prefix, \tpdastate_1), \eps, \pop(R_2), \tmp_1(\tau, R_2, \prefix)}$,
\item $\trule{\tmp_1(\tau, R_2, \prefix), \eps, \pop(R_1), \tmp_2(\tau, R_2, \prefix)}$,
\item $\trule{\tmp_2(\tau, R_2, \prefix), \eps, \push(R'), (\prefix, \tpdastate_2))}$ (for simulating the prefix),

\item $\trule{(\suffix, \tpdastate_1), a, \pop(R_2), \tmp_1(\tau, R_2, \suffix)}$,
\item $\trule{\tmp_1(\tau, R_2, \suffix), \eps, \pop(R_1), \tmp_2(\tau, R_2, \suffix)}$,
\item $\trule{\tmp_2(\tau, R_2, \suffix), \eps, \push(R'), (\suffix, \tpdastate_2))}$ (for simulating the suffix).
\end{itemize}

\subsubsection*{Timed Transitions}

For every state $\tpdastate \in \tpdastates$ and every pair of regions $R, R'$ 
such that $R' \in R^+[\reference\,\, \setto [0:0]]$ (this is a singleton set),
the set $\pdatrans$ contains the transition rules:

\begin{itemize}
\item $\trule{(\prefix, \tpdastate), \eps, \pop(R), \tmp(timed, \tpdastate, R, \prefix)}$,
\item $\trule{\tmp(timed, \tpdastate, R, \prefix), \eps, \push(R'), (\prefix, \tpdastate)}$.
\end{itemize}
Additionally, if $R' \vDash (\cclock \in [0:1))$, then $\pdatrans$ also contains the transitions

\begin{itemize}
\item $\trule{(\suffix, \tpdastate), \eps, \pop(R), \tmp(timed, R, \suffix)}$,
\item $\trule{\tmp(timed, R, \suffix), \eps, \push(R'), (\suffix, \tpdastate)}$.
\end{itemize}

\subsubsection*{Switching Modes}

In addition to the transitions described so far, $\pda$ must also be able to
switch from mode $\prefix$ to mode $\suffix$. This is done nondeterministically 
at any point in the simulation of the prefix. When the automaton changes mode, 
it resets the control clock $\cclock$.
For each state $\tpdastate \in \tpdastates$ and region $R$,
the set $\pdatrans$ contains the transition rules
$\trule{(\prefix, \tpdastate), \eps, \pop(R), \tmp(switch, \tpdastate, R)}$
and  $\trule{\tmp(switch, \tpdastate, R), \eps, \push(R'), (\suffix, \tpdastate)}$,
where $R'$ is the region in the singleton set $R[\cclock \setto 0]$.

\paragraph{Correctness.}

The simulation of the prefix (mode $\prefix$) works exactly like the simulation
in \cite{abdulla2012dense}. The simulation of the suffix (mode $\suffix$) only
imposes a restriction on the duration of the remaining run, namely that the
value of the control clock $\cclock$ may not reach 1.
In other words, the automaton may simulate any unit suffix.
Additionally, it reads an $a$ each time it simulates a discrete transition.
This, together with Lemma \ref{lemma:model1} implies the following result:

\begin{lemma}\label{lemma:main1}
There exists a zeno run in $\tpda$ if and only if
for the corresponding symbolic automaton $\pda$, we have $a^\omega \in Traces(\pda)$.
\end{lemma}

Using our construction, the size of $\pda$ is exponential in the size of $\tpda$.
The problem of checking $a^\omega \in Traces(\pda)$ is 
polynomial in the size of $\pda$ \cite{BEM97}.
This gives membership in {\sc ExpTime} for Theorem \ref{thm:exptime}.

\section{An {\sc Exptime} Lower Bound for the Zenoness Problem}

The following theorem  gives {\sc Exptime}-hardness  for the zenoness problem for TPDA (matching its  upper bound). 

\begin{theorem}
The zenoness problem for TPDA is {\sc ExpTime}-hard.
\end{theorem}

\begin{proof}
The following problem is {\sc ExpTime}-complete \cite{Heussner:2012:LMCS}:
Given a labelled pushdown automaton $\pda$ recognizing the language $L$ and 
$n$ finite automata $A_1, \dots, A_n$ recognizing languages $L_1, \dots, L_n$,
is the intersection $L \cap \bigcap_{i=1}^n L_i$ empty?
This problem can be reduced, in polynomial time, to the zenoness problem for a TPDA $\tpda$.
The pushdown part of $\tpda$ simulates $\pda$, while a clock $x_i$ encodes
the state of the finite automaton $A_i$. We can use an additional control clock to
ensure that no time passes during the simulation.
We may assume w.l.o.g. that the finite automata are free of $\eps$-transitions.
An $\eps$-transition of $\pda$ is simulated by the pushdown part of $\tpda$.
A labelled transition of $\pda$ is first simulated by the pushdown part of $\tpda$
and then followed by a sequence of transitions that checks and updates the clocks 
in order to ensure that each finite automaton $A_i$ is able to match the transition.

From a final state of $\pda$, we introduce a series of transitions that checks if all finite-state automata $A_i$
are also in their final states. If they are, we move to a special state of  $\tpda$ from which there exists a zeno run.
In this special state, we remove the restriction that time cannot pass and we add a self-loop performing a $\nop$ operation. Thus, the intersection $L \cap \bigcap_{i=1}^n L_i$
is empty if and only if $\tpda$ does not contain a zeno run.
\end{proof}

%%% Local Variables: 
%%% mode: latex
%%% TeX-master: "main"
%%% End: 

%%% Local Variables: 
%%% mode: latex
%%% TeX-master: "main"
%%% End: 

\section{Conclusion and Future Work} 

In this paper, we have considered the problem of detecting zeno runs in TPDA.
We showed that the zenoness problem for TPDA is {\sc ExpTime}-complete. 
The proof uses a reduction from the zenoness problem for TPDA to the 
problem of deciding whether $a^\omega$ is contained in the set of  traces of a PDA.
More specifically, given a TPDA $\tpda$, 
we construct a PDA $\pda$ which simulates zeno runs of $\tpda$ and whose size is exponential in the size of $\tpda$.

We are currently considering the problem of computing the minimal (or infimal, if it does not  exist) \emph{reachability cost}
in the model of \emph{priced} TPDA, in which discrete
transitions have \emph{firing costs} and stack contents have \emph{storage costs}, meaning that the cost of taking a timed transition
depends on the stack content.

Another interesting question is whether there are fragments of some suitable metric logic 
for which model checking TPDA is decidable.

\bibliographystyle{eptcs}
\bibliography{biblio}

\end{document}